\documentclass[twoside,11pt]{article}
\usepackage{amssymb,amsmath,amsthm,amsfonts, epsfig}           
\usepackage{mathtools} 
\usepackage{hyperref}
\usepackage{eucal}
\usepackage{a4wide}
\usepackage[svgnames]{xcolor}
\usepackage{mathptmx} 
\usepackage{graphicx}
\usepackage{tikz-cd}
 \usepackage{subfigure}
 \usepackage{fancyhdr}
 \usepackage{MnSymbol}
 \providecommand{\keywords}[1]{\small \textbf{Keywords:} #1}
  \providecommand{\msc}[1]{\small \textbf{Mathematics subject classification:} #1}

\title{Integrability of the $n$-dimensional axially symmetric Chaplygin sphere  \footnote{Dedicated to S.A. Chaplygin on the occasion of his 150th birthday.}}

\author{ Luis C.~Garc\'ia-Naranjo}

\numberwithin{equation}{section}
\numberwithin{table}{section}
\numberwithin{figure}{section}

\hypersetup{colorlinks=true,linkcolor=violet,citecolor=purple}

\newtheorem{theorem}{Theorem}[section]

\newtheorem{proposition}[theorem]{Proposition}
\newtheorem{corollary}[theorem]{Corollary}

{\theoremstyle{definition}

\newtheorem{remark}[theorem]{Remark}
\newtheorem*{remarks*}{Remarks}

}

\newtheorem{innercustomgeneric}{\customgenericname}
\providecommand{\customgenericname}{}
\newcommand{\newcustomtheorem}[2]{%
  \newenvironment{#1}[1]
  {%
   \renewcommand\customgenericname{#2}%
   \renewcommand\theinnercustomgeneric{##1}%
   \innercustomgeneric
  }
  {\endinnercustomgeneric}
}

\newcustomtheorem{customhyp}{}

\newcommand{\defn}[1]{{\bfseries\itshape{#1}}}

\def\headcolour{\color{Grey}}
\def\restr#1{\,\vrule height1.2ex width.4pt
  depth0.8ex\lower0.4ex\hbox{\scriptsize $\,#1$}}

\newcommand{\R}{\mathbb{R}}

\newcommand{\I}{\mathbb{I}}
\newcommand{\J}{\mathbb{J}}

\newcommand{\so}{\mathfrak{so}}

\newcommand{\SO}{\mathrm{SO}}
\newcommand{\SE}{\mathrm{SE}}

\newcommand{\Ss}{\mathrm{S}}

\begin{document}

\maketitle

\maketitle

\begin{abstract}
We consider the $n$-dimensional Chaplygin sphere  under the assumption that the mass distribution of  the
sphere is axisymmetric. 
 We prove that for initial conditions whose angular momentum about the contact point is vertical, the 
dynamics is quasi-periodic. For $n=4$ we perform the reduction by the associated $\SO(3)$ symmetry and show that
the reduced system is integrable by the Euler-Jacobi theorem.
  \end{abstract}

\vspace{0.2cm}

\msc{37J60, 70E18, 70E40, 58D19.}

{\small
\tableofcontents
}

\section{Introduction}

The Chaplygin sphere is perhaps the most interesting example of an integrable system in nonholonomic mechanics. It concerns 
 the motion of a sphere, whose centre of mass coincides with
its geometric centre, that rolls without slipping on the plane. The integrability of the problem was proved by S.A. Chaplygin in
his celebrated paper \cite{chapsphere}.

The $n$-dimensional generalisation of the problem was  introduced by Fedorov and Kozlov in~\cite{FedKoz} where the authors
conjecture that this generalisation is also integrable. To the author's best  knowledge such conjecture is only known to 
be true in the following cases:
\begin{enumerate}
\item[1.] the inertia tensor $\I:\so(n)\to \so(n)$  is spherical, i.e. a constant factor of the identity operator.  In this simple case the dynamics is trivially integrable since the angular velocity remains  constant along the motion;
\item[2.] the case treated by Jovanovi\'c in \cite{Jovan}. Here the initial condition is restricted to have 
  {\em  horizontal momentum}. Moreover,  the inertia operator $\I:\so(n)\to \so(n)$ is assumed to be of a very specific type and, in particular,
to map the space of rank-two matrices in $\so(n)$
into itself (see \cite[Eq. (49)]{Jovan}). 
\end{enumerate}

For $n\geq 4$ the inertia operator considered by Jovanovi\'c in case 2 above does not generally correspond
 to a physical inertia operator of
a multidimensional rigid body unless such body is axisymmetric (see \cite[Remark 2]{Jovan} and \cite[Appendix B]{FGNM18}).

In this paper we further analyse the dynamics of the system under the assumption that the distribution of mass of the sphere is axisymmetric. In our approach we take advantage of the associated additional $\SO(n-1)$ symmetry of the problem. This kind of 
symmetry analysis has already proved to be useful to determine new cases of integrability of the $n$-dimensional Veselova problem
\cite{FGNM18}.

Our main result is to prove that the dynamics of the problem for arbitrary $n$ 
 is quasi-periodic if the angular momentum about the contact point
is {\em vertical}. We also consider general initial conditions in the case  $n=4$   and show that  the reduction of the system by the additional 
$\SO(3)$ symmetry 
 is integrable by
the Euler-Jacobi theorem.

The paper is organised as follows. We first recall the equations of motion and their main properties in Section \ref{S:Prelim}.
Next, in Section \ref{S:VertHor}, we define the spaces of vertical and horizontal momentum both in 3 and $n$D.
Section \ref{S:Axisymmetric} studies the axisymmetric sphere for general $n$ and Section \ref{S:4D} focuses on the case $n=4$.

\section{Preliminaries}
\label{S:Prelim}

\subsection{The classical 3D Chaplygin sphere}

The homogeneity  of the plane where the
rolling takes place leads to a symmetry of the problem with respect to the action  of the euclidean 
 group $\SE(2)$. The  reduced equations of motion are  well-known and given by 
 \begin{equation}
 \label{eq:motion3D}
\dot M = M\times \Omega, \qquad \dot \gamma = \gamma \times \Omega,
\end{equation}
where $M \in \R^3$ is the angular momentum of the sphere about the contact point, $\Omega\in \R^3$ is the angular velocity,  $\gamma \in \R^3$ is the normal vector
to the plane and  `$\times$' denotes  the  cross product. 
All vectors $M, \Omega, \gamma$ are written in a
body frame that is attached to the centre of the sphere and satisfy
\begin{equation}
 \label{eq:mom3D}
M= \I(\Omega) + b\gamma \times ( \Omega \times \gamma),
\end{equation}
where the $3\times 3$ matrix $\I$ is the tensor of inertia, and $b=mr^2$ where $m$, $r>0$ denote the mass and the radius of the sphere.
We assume that the body frame is aligned with the principal axes of inertia so $\I = \mbox{diag}(I_1,I_2,I_3)$, with $I_j>0$ denoting
the principal moments of inertia.

The system~\eqref{eq:motion3D} possesses the trivial integral  $\|\gamma\|^2$ and from now
on we restrict our attention to  $\|\gamma\|^2=1$, and  interpret $\gamma$ as
an element of the unit sphere $\Ss^2$. For future reference we note that the first equation in \eqref{eq:motion3D} may  be equivalently 
written as
\begin{equation}
\label{eq:motion-velocity3D}
\I ( \dot \Omega) =  \I(\Omega)  \times \Omega -b \gamma \times (\dot \Omega \times \gamma).
\end{equation}

For a  fixed $\gamma\in \Ss^2$, Eq.~\eqref{eq:mom3D} defines a linear relation between $M$ and $\Omega$ that
we denote by $M=L_\gamma(\Omega)$. As may be checked directly, the  determinant of $L_\gamma$ satisfies 
\begin{equation*}
 \det (L_\gamma)  = c\left ( 1- b \langle {\tilde \I}^{-1}\gamma, \gamma \rangle \right )>0,
\end{equation*}
where ${\tilde \I}:= \I+b\mbox{Id}_3$,  the constant $c=\prod_{j=1}^3(I_j+b)=\det ({\tilde \I})$ and $\langle \cdot , \cdot \rangle$ is the euclidean
inner product in $\R^3$.
The equations of motion may be written down in  explicit form in terms of $(M,\gamma)\in \R^3\times \Ss^2$ by noting that the inversion of $L_\gamma$ leads to
\begin{equation}
\label{eq:invL3D}
\Omega= L_\gamma^{-1}(M)={\tilde \I}^{-1}M + \frac{ bc \langle {\tilde \I}^{-1} M , \gamma\rangle }{ \det (L_\gamma) }\tilde{\I}^{-1}\gamma.
\end{equation}

\subsection*{First integrals and measure preservation}

The equations of motion~\eqref{eq:motion3D} state that the expression of the vectors $M$ and $\gamma$ 
in the space frame is constant. This  observation is trivial in the case of  $\gamma$, but is quite  remarkable in the case of $M$, and 
leads to the existence of the following first integrals:
\begin{equation*}
f_1= \langle M,\gamma \rangle, \qquad f_2=\| M \|^2.
\end{equation*}
In addition, the system preserves the energy $H=\frac{1}{2}\langle M,\Omega \rangle$ and possesses the invariant measure
\begin{equation*}
\mu = \frac{1}{ \sqrt{\det  (L_\gamma) }} \, dM\, d\gamma = \sqrt{\det  (L_\gamma) } \, d\Omega\, d\gamma .
\end{equation*}
Therefore the equations~\eqref{eq:motion3D}, which define a vector field on the 5-dimensional phase
space $P=\R^3\times \Ss^2\ni(M,\gamma)$, possess 3 independent integrals $f_1, f_2, H$ together with the invariant measure
$\mu$, and are thus integrable by the Euler-Jacobi  theorem (see e.g. \cite{Arnold}). The explicit integration of the 
equations \eqref{eq:motion3D} was obtained by  Chaplyign
 in \cite{chapsphere}.

\subsection{The $n$-dimensional Chaplgyin sphere}

The multi-dimensional generalization of the Chaplygin sphere was first considered by 
Fedorov and Kozlov~\cite{FedKoz}. In this case the $n$-dimensional sphere rolls without slipping 
on an $n-1$ dimensional hyperplane whose homogeneity  leads to an $\SE(n-1)$ symmetry, and
the  reduced equations of the motion are given by
\begin{equation}
\label{eq:motion-general}
\dot M =[M , \Omega ], \qquad \dot \gamma = -\Omega \gamma,
\end{equation}
where the angular momentum about the contact point  $M$  and the angular velocity $\Omega$ are now 
elements of the Lie algebra $\so(n)$ of skew-symmetric matrices and $[M , \Omega ]$ denotes their commutator.
As before, 
$\gamma \in \R^n$ is the Poisson vector that gives body coordinates of the unit normal to the hyperplane where the rolling takes
place; in particular $\|\gamma\|=1$ throughout the motion so we think of $\gamma\in \Ss^{n-1}$. 
All quantities are expressed in the body frame that is attached to the centre of the sphere.
 The relationship between $M$ and $\Omega$ that generalizes \eqref{eq:mom3D}
  is\footnote{We recall that the wedge product of 
column vectors
$u, v\in \R^n$ is defined as $u\wedge v = uv^t-vu^t\in \so(n)$.}
\begin{equation}
\label{eq:M-O}
M=L_\gamma(\Omega)=  \I (\Omega) +b(\Omega \gamma) \wedge \gamma, 
\end{equation}
where, as before,  $b=mr^2$. The   inertia tensor is now a map $\I:\so(n)\to \so(n)$ of the form 
\begin{equation}
\label{eq:Inertia}
\I (\Omega) = \J \Omega + \Omega \J,
\end{equation}
where $\J$ is the \defn{mass tensor}, which is is a constant, symmetric, $n \times n$ matrix that depends on the mass distribution of the body.  By an appropriate choice of a body frame, $\J$ may be assumed to be diagonal with positive entries (see, e.g., Ratiu~\cite{Ratiu80}).

For further reference we note that, in analogy with~\eqref{eq:motion-velocity3D},  the first equation in~\eqref{eq:motion-general} may
be rewritten as
\begin{equation}
\label{eq:motion-velocitynD}
\I ( \dot \Omega) = [ \I(\Omega) , \Omega]  -b  (\dot \Omega \gamma)\wedge \gamma.
\end{equation}

To the author's best knowledge, an explicit expression for $L_\gamma^{-1}$ that generalizes \eqref{eq:invL3D} 
for a general inertia tensor is unknown. In Proposition~\ref{P:inversion} below we give such formula under the  assumption
that the   inertia tensor is axisymmetric.

\subsection*{First integrals and measure preservation}

As in the 3D case, the equations of motion~\eqref{eq:motion-general} state that the expressions of  $M$ and $\gamma$ 
in the space frame is constant and this leads to the existence of several integrals of motion. To see this, note that for any 
$\sigma \in \R$, the
matrix $M+\sigma \gamma \gamma^t$ undergoes an iso-spectral evolution:
\begin{equation}
\label{eq:iso-spectral}
\frac{d}{dt} (M+\sigma \gamma \gamma^t)= \left [ M+\sigma \gamma \gamma^t, \Omega \right ].
\end{equation}
As a consequence, the coefficients of the two-variable polynomial 
\begin{equation*}
\label{eq:pol-first-integrals}
p(\lambda, \sigma) = \det ( M+\sigma \gamma \gamma^t - \lambda \mbox{Id}_n),
\end{equation*}
are first integrals.  In addition to these integrals, the energy of the system is also preserved. In this $n$-dimensional case it is given
by
\begin{equation*}
H=\frac{1}{2} \langle M, \Omega \rangle_\kappa ,
\end{equation*}
where $\langle \cdot, \cdot \rangle_\kappa$ is the Killing metric in $\so(n)$ defined 
by $ \langle \xi_1, \xi_2 \rangle_\kappa = -\frac{1}{2}\mbox{tr}(\xi_1\xi_2)$ for $\xi_1, \xi_2 \in \so(n)$. 

The work of  Fedorov and Kozlov~\cite{FedKoz} shows that the $n$-dimensional system also possess a smooth invariant
measure that again may be written as 
\begin{equation}
\label{eq:invmeasuredet}
\mu = \frac{1}{ \sqrt{\det  (L_\gamma) }} \, dMd\gamma=  \sqrt{\det  (L_\gamma) } \, d\Omega\, d\gamma .
\end{equation}

Despite the large number of first integrals and the existence of an invariant measure, the 
integrability of the system for $n>3$ has only been established 
in a very particular case, described below,
by Jovanovi\'c~\cite{Jovan}, and for  spherical inertia tensors.

\section{Vertical and horizontal momentum}
\label{S:VertHor}

\subsection{The 3D case}
\label{SS:vert-hor-3D}

 The integration of the equations carried out by Chaplygin~\cite{chapsphere}  proceeds by first distinguishing two special classes of initial conditions
 that correspond to   \defn{vertical} and  \defn{horizontal momentum}. In  the first case the vectors
  $M$ and $\gamma$ are parallel and in the other perpendicular. With this in mind we define
 the subsets $\mathcal{V}, \mathcal{H},$ of the phase space $P=\R^3\times \Ss^2\ni(M,\gamma)$  by:
 \begin{equation}
 \label{eq:VHgeneraln}
\mathcal{V} = \{ (M,\gamma) \in P \, : \, M\times \gamma =0\}, \qquad  \mathcal{H} = \{ (M,\gamma) \in P \, : \, 
f_2=\langle M, \gamma \rangle=0\}.
\end{equation}
These are   submanifolds of $P$, of dimension
3 and 4 respectively, which are invariant by the flow of~\eqref{eq:motion3D}. Their intersection $\mathcal{V}\cap \mathcal{H}$ consists of  initial
conditions having $M=0$, which corresponds to the  sphere being at rest.
 
In his celebrated work~\cite{chapsphere},  Chaplygin integrated  the equations of motion for initial conditions in $\mathcal{H}$
in terms of hyperelliptic functions on a  genus 2 Riemann surface. He then  showed that the integration for generic initial 
conditions that do not belong to $\mathcal{H}$ nor $\mathcal{V}$, may be reduced to the case of  horizontal momentum by means
of an insightful   change of variables.

The integration of the equations for vertical momentum is much simpler since
 the equations of motion in this case reduce
to the standard Euler equations for the motion of a rigid body with tensor of inertia ${\tilde \I}$.
We recall this well-known result in the following proposition.

\begin{proposition}
\label{P:EulerTop}
Denote by $\Omega(t)$  the angular
velocity along a solution of \eqref{eq:motion3D} whose initial condition $(M_0,\gamma_0)$ lies in $\mathcal{V}$. Then $\Omega(t)$
is a solution of the Euler equations
 \begin{equation}
 \label{eq:vert-3D}
{\tilde \I} \dot \Omega = ({\tilde \I}\Omega) \times \Omega.
\end{equation}
\end{proposition}
\begin{proof}
Denote by $(M(t),\gamma(t))$ such solution.  Since $\mathcal{V}$  is invariant then 
 $M(t)=\lambda \gamma(t)$ for a scalar  $\lambda\in \R$ that is necessarily constant since it 
satisfies $f_1(M_0,\gamma_0)=f_1(M(t),\gamma (t))=\lambda^2$. As a consequence we have $2H=\lambda \langle \Omega (t), \gamma (t) \rangle$, which shows
that $\langle \Omega (t), \gamma (t) \rangle$ is also constant along the motion. Therefore we have $\frac{d}{dt} \langle \Omega (t), \gamma (t) \rangle =
 \langle \dot \Omega (t) , \gamma (t) \rangle =0$ and  \eqref{eq:motion-velocity3D} may be rewritten as
 \eqref{eq:vert-3D}.
\end{proof}

\subsection{The $nD$-case}

The discussion above about the vertical and horizontal momentum spaces 
may be generalized to $n$-dimensions by defining  $\mathcal{H}$ and $\mathcal{V}$ as
 the following submanifolds of the phase space $P=\so(n)\times \Ss^{n-1}$: 
 \begin{equation}
 \label{eq:def-H-V}
\mathcal{V} = \{(M,\gamma)\in P \, : \ M\gamma =0 \}, \qquad \mathcal{H} = \{(M,\gamma)\in P \, : \, M-(M\gamma)\wedge \gamma  =0 \}.
\end{equation}
It may be easily checked using the equations of motion \eqref{eq:motion-general} that $\mathcal{V}$ and $\mathcal{H}$
  are indeed invariant under the flow. The conditions to belong to $\mathcal{H}$ and $\mathcal{V}$ are that the 
{\em space} representation of the angular momentum about the contact point  has the respective form:
\begin{equation*}
 \begin{pmatrix} \so(n-1) 
&  {\bf 0}_{(n-1)\times 1}  \\   {\bf 0}_{ 1 \times (n-1)}  & 0 \end{pmatrix}, \qquad \begin{pmatrix} {\bf 0}_{(n-1)\times (n-1)} & \R^{n-1} \\ -(\R^{n-1} )^T  & 0 \end{pmatrix} ,
\end{equation*}
where we have assumed that $e_n=(0,\dots, 0,1)^t$ is the  normal vector to the $(n-1)$-dimensional hyperplane where the rolling takes
place. 
Considering that $\gamma\in \Ss^{n-1}$, it follows that the dimension of  $\mathcal{V}$ is $\frac{n(n-1)}{2}$  whereas that of $\mathcal{H}$ is $2(n-1)$.
In particular, and in contrast  to the case $n=3$,  the set of vertical momentum is much bigger than that of horizontal momentum
if $n$ is large.

As mentioned in the introduction, the only known results of integrability for $n\geq 4$,  were given by Jovanovi\'c \cite{Jovan}. 
His work is concerned only with 
 initial conditions on $\mathcal{H}$, and assumes that the inertia tensor is of a very specific form (see \cite[Eq. (49)]{Jovan}). 
If one requires this inertia tensor to be physical, i.e. satisfies condition~\eqref{eq:Inertia}, this
leads to the  condition that the sphere is axisymmetric which we treat below and  is the main topic of this paper.

\begin{remark} The argument in the proof of Proposition~\ref{P:EulerTop}, that shows that the 3D Chaplygin sphere evolves as
the Euler top with inertia tensor $\tilde \I$ along the vertical space, depends crucially  on the 
assumption that $n=3$ and may not be extended for $n>3$. In fact,  at the end of Section~\ref{S:Axisymmetric} below, we show
 that such simplification is not
possible for $n>3$.
\end{remark}

\section{The axisymmetric Chaplygin sphere}
\label{S:Axisymmetric}

Suppose now that the mass distribution of the sphere is axisymmetric. 
We choose the body frame $\{E_1, \dots, E_n\}$  in such way that  $E_n$ is aligned with the axis of symmetry. 
With the
appropriate  normalization of units this leads to the following condition for the mass matrix:
\begin{equation}
\label{eq:axisymm}
\J =\mbox{Id}_n+a E_nE_n^t,
\end{equation}
for a real parameter $a$ that for  physical reasons is assumed to satisfy $-1\leq a\leq 1$. The case $a=0$ corresponds to a spherical inertia tensor.
We note that in the 3D case, our assumptions imply that the $3\times 3$  inertia matrix is the diagonal matrix with
entries  $(2+a,2+a,2)$.

 With our assumption \eqref{eq:axisymm} we  may write $\I (\Omega) =2\Omega +a (\Omega E_n) \wedge E_n$, and \eqref{eq:M-O}
becomes
\begin{equation}
 \label{eq:M-O-axi-symm}
M=L_\gamma(\Omega)= 2\Omega + a (\Omega E_n) \wedge E_n + b(\Omega \gamma) \wedge \gamma.
\end{equation}
We may also rewrite Eq.~\eqref{eq:motion-velocitynD} as
\begin{equation*}
2\dot \Omega + a(\dot \Omega E_n)\wedge E_n + b (\dot \Omega \gamma )\wedge \gamma = a [ (\Omega E_n)\wedge E_n , \Omega].
\end{equation*}
This equation shows that if the inertia tensor is spherical ($a=0$) then $\dot \Omega =0$ and the angular velocity is constant along the
motion. This observation was already made in~\cite{HGN09}.

 For the rest of the paper we write $x:=\gamma_n=\langle \gamma, E_n \rangle$.
We also denote
\begin{equation}
\label{eq:Delta}
 \Delta(x) = (2+a)(2+b)-abx^2,
\end{equation}
and we note that $\Delta(x)>0$ due to  the restrictions that $x^2\leq 1$ and $a\in [-1,1]$. 

\begin{proposition}
\label{P:Measure}
Under the axisymmetric assumption~\eqref{eq:axisymm} the invariant measure of the multi-dimensional 
Chaplygin sphere is given by
\begin{equation}
\mu = \frac{1}{\Delta(x)^{\frac{n-2}{2}}}\,dM\, d\gamma = \Delta(x)^{\frac{n-2}{2}}\,d\Omega\, d\gamma.
\end{equation}
\end{proposition}
\begin{proof}
We will prove that $\det(L_\gamma)$ is proportional to $\Delta(x)^{n-2}$. The result then follows from the 
formula~\eqref{eq:invmeasuredet} for the invariant measure given by Fedorov and Kozlov~\cite{FedKoz}.

Suppose that $\gamma$ and $E_n$ are linearly independent and let $\{\gamma, E_n, w_1, \dots , w_{n-2}\}$ be a basis
of $\R^n$ with the property that $\langle w_j, \gamma\rangle = \langle w_j, E_n \rangle =0$ for all $j$. Using 
the expression \eqref{eq:M-O-axi-symm} for the for the linear operator  $L_\gamma:\so(n)\to \so(n)$, one computes
\begin{equation*}
\begin{split}
& L_\gamma( E_n \wedge w_j) = (2+a) E_n \wedge w_j + bx \gamma \wedge w_j, \qquad
 L_\gamma( \gamma \wedge w_j) =ax E_n \wedge w_j + (2+b) \gamma \wedge w_j, \\
 &L_\gamma( E_n \wedge \gamma ) = (2+a+b)E_n \wedge \gamma , \qquad  \qquad \qquad L_\gamma( w_i \wedge w_j) =2w_i \wedge w_j. \\ 
\end{split}
\end{equation*}
Therefore, the matrix representation of $L_\gamma$ with respect to the ordered basis 
\begin{equation*}
 \; E_n \wedge w_1, \; \gamma \wedge w_1,  \; \dots \; , E_n \wedge w_{n-2}, \; \gamma \wedge w_{n-2}, \;E_n \wedge \gamma, \;
w_i\wedge w_j,
\end{equation*}
of $\so(n)$, is given in block-diagonal form as
\begin{equation*}
\begin{pmatrix} C_1 & 0 & 0 &  \cdots & 0 \\ 0 & \ddots  & 0 & \cdots  & \vdots \\ 0 & \cdots & C_{n-2} &0 & 0 \\  0 & \cdots & 0 & 2+a+b & 0 & \\
0 & \cdots & 0 & 0 & 2 \mbox{Id}_{(n-2)(n-3)/2}
\end{pmatrix} ,
\end{equation*}
where the $2\times 2$ blocks $C_j$, $j=1,\dots , n-2$,  are all identical and equal to 
\begin{equation*}
C=\begin{pmatrix} 2+a  & ax \\ bx & 2+b \end{pmatrix}.
\end{equation*}
Given that $\det(C)=\Delta(x)$, it follows that, up to the constant factor $(2+a+b)2^{(n-2)(n-3)/2}$, the determinant of $L_\gamma$ equals $\Delta(x)^{n-2}$. The proof for configurations where $\gamma$ and $E_n$ are parallel follows by continuity.
\end{proof}

The following proposition gives
the explicit form of $L_\gamma^{-1}$ under our axisymmetric assumption~\eqref{eq:axisymm}.

\begin{proposition}
\label{P:inversion}
Equation   \eqref{eq:M-O-axi-symm} may be inverted to express  $\Omega$ in terms of  $M$ and $\gamma$
 as
\begin{equation}
\Omega =L_\gamma^{-1}(M)=  \frac{1}{2}M + \frac{1}{2\Delta(x)}K, 
\end{equation}
where 
\begin{equation}
\begin{split}
K= & abx (M\gamma) \wedge E_n -a(2+b)(ME_n)\wedge E_n  + abx(ME_n)\wedge \gamma
-b(2+a) (M\gamma) \wedge \gamma   \\ 
& \qquad \qquad -\frac{ab(4+a+b) \langle M\gamma ,E_n \rangle }{2+a+b} \gamma \wedge E_n.
\end{split}
\end{equation}
\end{proposition}
\begin{proof}
We will obtain expressions for $(\Omega E_n) \wedge E_n$ and  $(\Omega \gamma) \wedge \gamma$
in terms of $M$, and $\gamma$. Once  this is done,  the inversion of \eqref{eq:M-O-axi-symm} is trivial.

First note that, thanks to the skew-symmetry of $\Omega$, multiplying \eqref{eq:M-O-axi-symm} on the left by $\gamma^t$ and on the right by
$E_n$  leads to 
\begin{equation}
\label{eq:aux-inversion}
\langle \Omega \gamma, E_n \rangle = \frac{\langle M\gamma, E_n \rangle}{2+a+b}.
\end{equation}
Next, multiplying  \eqref{eq:M-O-axi-symm} on the right by $E_n$ (respectively $\gamma$) and taking
the exterior product of the resulting expression with $E_n$ leads to the system of two linear equations
for the unknowns $U:=(\Omega E_n)\wedge E_n$ and $Z:=(\Omega \gamma) \wedge E_n$:
\begin{equation*}
(2+a)U +bxZ = (ME_n) \wedge E_n + b \frac{\langle M\gamma, E_n \rangle}{2+a+b} \gamma \wedge E_n, \qquad
ax U + (2+b) Z = (M\gamma ) \wedge E_n,
\end{equation*}
where we have made use of the expression for $(\Omega \gamma, E_n)$ given above.
 The determinant of this linear system  is precisely $\Delta(x)$ in~\eqref{eq:Delta} and its unique solution gives
 \begin{equation*}
(\Omega E_n)\wedge E_n=U=  \frac{1}{\Delta (x)} \left ( (2+b) (ME_n)\wedge E_n -bx (M\gamma) \wedge E_n + \frac{b(2+b)\langle M\gamma , E_n\rangle}{2+a+b} \gamma \wedge E_n \right ).
\end{equation*}

Proceeding in  a completely analogous manner, multiplying  \eqref{eq:M-O-axi-symm} on the right by $E_n$ (respectively $\gamma$) and taking
the exterior product of the resulting expression with $\gamma$ one obtains a linear system
of two equations for the unknowns $(\Omega \gamma)\wedge \gamma$ and $(\Omega E_n) \wedge \gamma$, whose solution gives
\begin{equation*}
(\Omega \gamma) \wedge \gamma=   \frac{1}{\Delta (x)} \left ( -ax (ME_n)\wedge \gamma +(2+a) (M\gamma) \wedge \gamma + \frac{a(2+a)\langle M\gamma , E_n \rangle}{2+a+b} \gamma \wedge E_n \right ).
\end{equation*}
The proof of the proposition follows by inserting the above expressions in \eqref{eq:M-O-axi-symm}.
\end{proof}

\subsection{The additional $\SO(n-1)$ symmetry} 
\label{SS:extra-symmetry}
Our assumption that the body is axisymmetric leads to a  symmetry of the equations that we now describe. We
shall write
\begin{equation}
\label{eq:SOnminus1}
\SO(n-1)= \{ h \in \SO(n) \, : \,  hE_n=E_n\}.
\end{equation}
Then $\SO(n-1)$ acts on the phase space $P=\so(n)\times \Ss^{n-1}\ni (M,\gamma)$ by
\begin{equation}
\label{eq:action}
h\cdot (M, \gamma) = (hMh^{-1}, h\gamma).
\end{equation}
Using the condition that $ hE_n=E_n$ one may use Proposition~\ref{P:inversion} to check that $\Omega$ is mapped into  $h\Omega h^{-1}$ by the action of $h\in \SO(n-1)$.
 Hence, taking into account the equivariance of 
  the matrix commutator with respect to conjugation, 
  it is straightforward to check that the system~\eqref{eq:motion-general} is
  $\SO(n-1)$-equivariant and the dynamics may be reduced to $P/\SO(n-1)$.

\subsection{The solution in the case of vertical momentum}

In this section we show that for an axisymmetric sphere of arbitrary dimension $n$, 
the system may be explicitly integrated for 
initial conditions on the space of vertical angular momentum $\mathcal{V}$ defined by~\eqref{eq:def-H-V}.
These solutions are quasi-periodic in $P=\so(n)\times \Ss^{n-1}$ and are relative equilibria of the $\SO(n-1)$ action described above.

For our purposes it is  convenient to note that the vertical space   $\mathcal{V}$ defined by~\eqref{eq:def-H-V} may be
equivalently described in terms of $\gamma$ and $\Omega$ as:
\begin{equation}
\label{eq:V-velocity}
\mathcal{V} = \{ (\Omega, \gamma) \in \so(n) \times \Ss^{n-1}\, : \, (2+b)\Omega \gamma + a x \Omega E_n =0 \quad \mbox{and} \quad
\langle \Omega \gamma, E_n \rangle=0 \}.
\end{equation}
This may be checked by multiplying Eq. \eqref{eq:M-O-axi-symm} on the right by $\gamma$ and using~\eqref{eq:aux-inversion} to conclude that 
$\langle \Omega \gamma, E_n \rangle=0$ along $\mathcal{V}$.

\subsection*{The 3D case}
Let us first consider the case $n=3$.
Under our assumptions on the inertia tensor, we have $$\tilde \I=\mbox{diag}(2+a+b,2+a+b, 2+b),$$
and \eqref{eq:vert-3D} yields
\begin{equation*}
\dot \Omega_1 = \frac{a\Omega_2\Omega_3}{2+a+b}, \qquad 
\dot \Omega_2 = - \frac{a\Omega_1\Omega_3}{2+a+b}, \qquad
\dot \Omega_3=0.
\end{equation*}
The  solution of this system with initial condition $\Omega(0)= \Omega_0\in \R^3$ is easily obtained
 in terms of sines
and cosines.
In order to compare with the $n$D case ahead, we write it in matrix form as
\begin{equation*}
\hat \Omega (t) = \exp (\zeta t)  \hat \Omega_0\exp (-\zeta t), 
\end{equation*}
where
\begin{equation}
\label{eq:xiV3D}
 \zeta = \frac{-a}{2+a+b}  \begin{pmatrix} 
0 & - \bar \Omega_3 &0 \\
\bar \Omega_3& 0 &0\\ 
0 & 0  & 0
\end{pmatrix} =  \frac{-a}{2+a+b}\left ( \hat \Omega_0  -   (\hat \Omega_0 E_3)\wedge E_3 \right ).
\end{equation}
Here $\bar \Omega_3$ denotes the third component of $\Omega_0$ and 
$\; \hat{} : \R^3 \to \so(3)$ is the `hat map' (see e.g.~\cite{MaRa}) which is the well-known Lie algebra 
isomorphism determined by the condition that
$\hat u v=u\times v$ for all $u,v\in \R^3$. 

We now claim that the solution of the Poisson equation $\dot \gamma = \gamma \times \Omega(t)$  is
\begin{equation*}
 \gamma(t)=  \exp (\zeta t)\gamma_0,
\end{equation*}
where $\gamma_0$ is the initial condition. This is easy to prove after noting that, because our initial condition lies
in  $\mathcal{V}$, 
the  matrix $\zeta$ in \eqref{eq:xiV3D} satisfies
\begin{equation*}
\zeta \gamma_0 = -\hat \Omega_0  \gamma_0 = \gamma_0 \times \Omega_0.
\end{equation*}
Indeed, this may be checked by using  the 3D version of~\eqref{eq:V-velocity}. Finally, in view of~\eqref{eq:M-O-axi-symm}
and since $\zeta E_3=0$, we have
\begin{equation*}
\hat M (t) = \exp (\zeta t)  \hat M_0\exp (-\zeta t),
\end{equation*}
where $M_0$ is the initial condition.

The above 
  expressions together with the definition of the $\SO(2)$ action defined by~\eqref{eq:SOnminus1} 
and \eqref{eq:action}, imply  
that the solutions along  $\mathcal{V}$ are contained in the orbits of the  $\SO(2)$ action. In other words, they
are \defn{relative equilibria}.

\subsection*{The $n$D case}

We now generalize the discussion above for general $n$ by showing that the dynamics along $ \mathcal{V}$ consists of relative equilibria with 
respect to the  $\SO(n-1)$ action defined by~\eqref{eq:SOnminus1} 
and~\eqref{eq:action}.
As we shall see,  for $n\geq 4$ the expression for the `velocity' $\zeta$ of the relative equilibria
 is more intricate than  \eqref{eq:xiV3D}. Since $\SO(n-1)$  is compact, the corresponding solutions on $P$ are 
quasi-periodic on tori whose generic dimension is $\mbox{rank}( \SO(n-1))= [\frac{n-1}{2} ]$.

\begin{theorem}
\label{Th:Re}
Under the axisymmetric assumption~\eqref{eq:axisymm} on the inertia tensor,
the solution of the Chaplygin sphere equations \eqref{eq:motion-general} with initial condition  $(M_0,\gamma_0)\in \mathcal{V}$ is quasi-periodic and given by
\begin{equation*}
M(t) = \exp(\zeta t) M_0 \exp(-\zeta t), \qquad \gamma(t)=  \exp(\zeta t) \gamma_0,
\end{equation*}
where
\begin{equation*}
\zeta = -\frac{a(2+b(1-x_0^2))}{2\Delta (x_0)}M_0  + \frac{1}{ 2\Delta(x_0)} \left ( a(2+b) (M_0 E_n) \wedge E_n
-abx_0 (M_0E_n) \wedge \gamma_0 \right ).
\end{equation*}
Here  $x_0=\langle E_n,\gamma_0\rangle$ denotes 
the $n^{th}$ component of $\gamma_0$. In particular $x(t)$   is constant
and equal to $x_0$ along the motion and the dynamics along $ \mathcal{V}$ consists of relative equilibria with 
respect to the  $\SO(n-1)$ action defined by~\eqref{eq:SOnminus1} 
and~\eqref{eq:action}. 
\end{theorem}

\begin{proof}
That $x(t)$ is constant follows by direct differentiation of $x=\langle \gamma, E_n \rangle$ and the fact that $\langle \Omega \gamma, E_n \rangle =0$ on $\mathcal{V}$.
 
Next note that a direct calculation gives $\zeta E_n=0$, which implies that $\zeta$ belongs to the Lie algebra of the 
group $\SO(n-1)$ defined by~\eqref{eq:SOnminus1}.
On the other hand, 
the condition $M_0\gamma_0=0$ together with  Proposition~\ref{P:inversion}
shows that the initial angular velocity $\Omega_0$ satisfies
\begin{equation*}
\Omega_0 = -\zeta + \frac{(2+b)}{\Delta (x_0)}M_0.
\end{equation*}
Therefore 
\begin{equation*}
[M_0, \Omega_0]= [\zeta ,M_0], \qquad \mbox{and} \qquad  \zeta \gamma_0 = -\Omega_0 \gamma_0.
\end{equation*}
This shows that the velocity vector $(\dot M(0), \dot \gamma(0))$ of the curve $(M(t),\gamma(t))$  in the statement of the
theorem coincides with the vector field  defined by  the Chaplygin sphere equations \eqref{eq:motion-general}
at the initial condition $(M_0,\gamma_0)$. Considering that these equations are equivariant and that 
$(M(t),\gamma(t))$ coincides with the 
 $\SO(n-1)$ action restricted to the 1-parameter subgroup $t\mapsto \exp(\zeta t)$ of
 $\SO(n-1)$, it follows that $(M(t),\gamma(t))$ is an integral curve of \eqref{eq:motion-general} (see e.g. \cite[Lemma 4.2.1.1]{CDS-book}).
\end{proof}

\begin{corollary}
Under the axisymmetric assumption~\eqref{eq:axisymm} on the inertia tensor and for initial conditions on $ \mathcal{V}$,
the angular velocity $\Omega(t)$   is also quasi-periodic and given by
\begin{equation*}
\Omega(t) = \exp(\zeta t) \Omega_0 \exp(-\zeta t), 
\end{equation*}
where $\Omega_0$ is the initial angular velocity.
\end{corollary}
\begin{proof}
This follows from the above theorem by noting that the action~\eqref{eq:action} maps $\Omega \mapsto h\Omega h^{-1}$.
\end{proof}

The matrix $\zeta$  in the statement of Theorem~\ref{Th:Re} may be written in terms of the initial velocity 
$\Omega_0$ and 
$\gamma_0$. A direct calculation using \eqref{eq:M-O-axi-symm} and the condition $0=M_0\gamma_0=\Omega_0((2+b)\gamma_0+ax_0E_n)$, that is valid on $\mathcal{V}$, 
gives
\begin{equation}
\label{eq:xiVnD}
\zeta=-\frac{a(2+b(1-x_0^2))}{\Delta (x_0)}(\Omega_0-(\Omega_0E_n)\wedge E_n) 
-\frac{abx_0}{\Delta(x_0)}(\Omega_0E_n)\wedge 
(\gamma_0-x_0E_n).
\end{equation}

The following  proposition shows that our treatment of the $n$-dimensional case is consistent with the results that
were obtained above for $n=3$.

\begin{proposition}
\label{P:n3special}
When $n=3$ the expressions for $\zeta$ in Eq.~\eqref{eq:xiVnD} and Eq.~\eqref{eq:xiV3D} coincide.
\end{proposition}
\begin{proof}
First suppose that the initial condition has $x_0=0$. Then, because of~\eqref{eq:V-velocity} we have $0=\hat\Omega_0 \gamma_0
= \Omega_0\times \gamma_0$. In other words, the vectors $\Omega_0$ and $\gamma_0$ are parallel. Considering that
the last entry of $\gamma_0$ vanishes, the same is true about $\bar \Omega_3$ in~\eqref{eq:xiV3D}. Thus, in this case
$ \hat \Omega_0-(\hat \Omega_0E_3)\wedge E_3=0$ and both expressions for $\zeta$ equal $0$.

Now suppose that  $\hat \Omega_0E_3=0$. Because of~\eqref{eq:V-velocity} we have again $0=\hat\Omega_0 \gamma_0
= \Omega_0\times \gamma_0$ and the same reasoning  leads to the conclusion that both expressions for  $\zeta$ vanish.

Suppose now that $x_0\neq0$ and  $\hat \Omega_0E_3\neq 0$, and note that both matrices, $\hat \Omega_0-(\hat \Omega_0E_3)\wedge E_3$ and $(\hat \Omega_0E_3)\wedge 
(\gamma_0-x_0E_3)$, annihilate $E_3$. Given that any two matrices in $\so(3)$ with the same null-vector are proportional we have
\begin{equation*}
\mu (\hat  \Omega_0- (\hat \Omega_0E_3)\wedge E_3) ) = (\hat \Omega_0E_3)\wedge 
(\gamma_0-x_0E_3),
\end{equation*}
for a certain $\mu \in \R$. Multiplying the above equation by $\gamma_0$ on the right and using~\eqref{eq:V-velocity}  gives
 \begin{equation*}
\left ( -\mu x_0 \frac{a+b+2}{2+b}-1+x_0^2 \right ) \hat \Omega_0E_3 =0.
\end{equation*}
Because of our assumption that $x_0\neq0$ and  $\hat \Omega_0E_3\neq 0$ we conclude that 
$\mu=- \frac{(1-x_0^2)(2+b)}{x_0(2+a+b)}$ and hence
\begin{equation*}
 (\hat \Omega_0E_3)\wedge 
(\gamma_0-x_0E_3) =  -\frac{(1-x_0^2)(2+b)}{x_0(2+a+b)}  (\hat  \Omega_0- (\hat \Omega_0E_3)\wedge E_3) ).
\end{equation*}
Substitution of the above expression in Eq.~\eqref{eq:xiVnD}  simplifies to Eq.~\eqref{eq:xiV3D}. 
\end{proof}

\subsection*{Comparison with the solutions of the $n$-dimensional free rigid body}

Proposition~\ref{P:EulerTop} showed that for the 3-dimensional Chaplygin sphere, the solutions along the invariant vertical space
$\mathcal{V}$ are also solutions of the  Euler equations for a free rigid body with inertia tensor $\tilde \I=  \I + b\mbox{Id}_3$. It is 
natural to ask if  such property is also valid in the multi-dimensional case. We shall  prove that, although the vertical space
$\mathcal{V}$ is invariant by the the flow of both systems, their  solutions  are generically distinct.

In order to compare the multi-dimensional solutions of the two systems, 
consider the
equations of motion of a multi-dimensional rigid body, accompanied  with the evolution equation of the Poisson vector $\gamma$:
\begin{equation}
\label{eq:Manakov}
\tilde \I (\dot \Omega)= [ \tilde \I(\Omega), \Omega], \qquad \dot \gamma =-\Omega \gamma, \qquad \Omega\in \so(n), \; \gamma\in  \Ss^{n-1}.
\end{equation}
The inertia tensor $\tilde \I := \I +b\mbox{Id}_{\so(n)}$ and we continue to assume that  $\I$ is defined in terms of the mass matrix in \eqref{eq:axisymm}.
We also continue to denote $x:=\gamma_n=\langle \gamma, E_n\rangle$.

\begin{proposition}\footnote{The inclusion of this result  in the final version of the paper was suggested  by one of the  anonymous referees.}
The flow of~\eqref{eq:Manakov} leaves the  vertical space $\mathcal{V}$ defined by~\eqref{eq:V-velocity}  invariant.
\end{proposition}
\begin{proof}
It is a simple exercise to check that the set where $\tilde \I(\Omega) \gamma =0$  is invariant by the flow of~\eqref{eq:Manakov}. But 
this set exactly coincides with $\mathcal{V}$ given by~\eqref{eq:V-velocity}.
Indeed, due to our assumptions on the inertia tensor $\tilde \I$ we have
\begin{equation}
\label{eq:itildendim}
\tilde \I(\Omega)= (2+b)\Omega +a (\Omega E_n)\wedge E_n,
\end{equation}
and it follows that
\begin{equation*}
\tilde \I(\Omega) \gamma  = (2+b) \Omega \gamma +ax \Omega E_n +a\langle \Omega \gamma , E_n\rangle E_n, \qquad \langle \tilde \I (\Omega) 
\gamma , E_n \rangle = (2+a+b) \langle \Omega \gamma , E_n\rangle.
\end{equation*}
Therefore $\tilde \I(\Omega) \gamma =0$ if and only if  $(\Omega, \gamma ) \in \mathcal{V}$ as defined  by~\eqref{eq:V-velocity}.
\end{proof}

The solutions of~\eqref{eq:Manakov} along $\mathcal{V}$ are given by the following.
\begin{proposition}
For the inertia tensor $\tilde \I$ given by~\eqref{eq:itildendim},
the solution of the multidimensional rigid body equations~\eqref{eq:Manakov} with initial condition  $(\Omega_0,\gamma_0)\in \mathcal{V}$ is quasi-periodic and given by
\begin{equation*}
\Omega(t) = \exp(\chi t) \Omega_0  \exp(-\chi t), \quad \gamma(t)=  \exp(\chi t) \gamma_0,
 \end{equation*}
where
\begin{equation}
\label{eq:solManakov}
\chi = \frac{-a}{2+a+b}\left ( \Omega_0 - (\Omega_0E_n)\wedge E_n \right ).
\end{equation}
\end{proposition}
\begin{proof}
Just like the axisymmetric Chaplygin sphere, the system~\eqref{eq:Manakov} 
is equivariant with respect to the action $h\cdot (\Omega, \gamma)\mapsto (h\Omega h^{-1}, h\gamma)$ where  $h\in \SO(n-1)$ as defined
by \eqref{eq:SOnminus1}. Moreover,  $\chi$ belongs to the Lie algebra  of $\SO(n-1)$ since $\chi E_n=0$. Considering that
\begin{equation*}
\Omega_0=  -\chi + \frac{1}{2+a+b} \tilde \I (\Omega_0),
\end{equation*}
we have $[\tilde \I (\Omega_0) ,\Omega_0]=[\chi , \tilde \I (\Omega_0) ]$ for all $\Omega_0\in \so(n)$, and $-\Omega_0\gamma_0=\chi \gamma_0$
for $(\Omega_0,\gamma_0)\in \mathcal{V}$. The rest of the proof proceeds as in the proof of Theorem~\ref{Th:Re}.
\end{proof}

For $n\geq 4$ the expression for $\zeta$ in~\eqref{eq:xiVnD} does not simplify to agree with  $\chi$ given by~\eqref{eq:solManakov}
(compare with Proposition~\ref{P:n3special} for $n=3$). Moreover,
it is possible to find initial conditions  $(\Omega_0,\gamma_0)\in \mathcal{V}$ with the property that $[\chi-\zeta, \Omega_0]\neq 0$.
For these initial conditions we have
\begin{equation*}
 \exp(\zeta t) \Omega_0  \exp(-\zeta t) \neq  \exp(\chi t) \Omega_0  \exp(-\chi t)
\end{equation*}
showing that the solutions for both systems are generally distinct.

\section{The 4-dimensional Chaplygin sphere}
\label{S:4D}

Now consider in more detail the special case $n=4$.  We briefly describe some facts that are valid for general inertia tensors
and then go back to our discussion of the  axisymmetric case.
In our analysis we shall write
\begin{equation*}
 M = \begin{pmatrix} \hat k & u \\ -u^T & 0 \end{pmatrix}, \quad  \Omega = \begin{pmatrix} \hat \eta & \xi \\ -\xi^T & 0 \end{pmatrix}\in \so(4), \quad \gamma = (q,x)^T \in \R^4,
\end{equation*}
where $x\in [-1,1]$, and $k,u, \xi, \eta,q \in \R^3$ are column vectors. 
The equations  of motion \eqref{eq:motion-general} rewrite in our notation as:
\begin{equation}
\label{eq:motion-n4}
\dot k = k\times \eta+ u \times \xi,  \quad \dot u = k\times \xi + u\times \eta, \quad \dot q  =q \times \eta -x\xi,  \quad \dot x  = \langle q,\xi
\rangle.
\end{equation}
The geometric integral is $\|q\|^2+x^2=1$ and the phase space $P=\R^3\times \R^3 \times \Ss^3 \ni (k, u, (q,x))$ is 9-dimensional.

The first integrals arising from the conservation of angular momentum about the contact point may be written down explicitly as:
\begin{equation*}
F_1= \langle k,u \rangle, \quad
F_2= \|k\|^2+\|u\|^2, \quad F_3=\| u\times q\|^2+x^2\|k\|^2+\langle k,q \rangle^2-2x\langle u,k\times q\rangle.
\end{equation*}
Indeed, these functions satisfy
\begin{equation*}
p(\lambda, \sigma) = \det ( M+\sigma \gamma \gamma^t - \lambda \mbox{Id}_4) = \lambda^4 -\lambda^3\sigma \|\gamma\|^2 +\lambda^2F_2 +\lambda \sigma F_3 -F_1^2,
\end{equation*}
and the coefficients of this polynomial in $(\lambda, \sigma)$ are first integrals because of the iso-spectral evolution \eqref{eq:iso-spectral}. On the other hand, the energy integral  writes as:
\begin{equation*}
H=\frac{1}{2} ( \langle u,\xi \rangle +\langle k,\eta\rangle).
\end{equation*}

The invariant sets \eqref{eq:VHgeneraln} of horizontal $\mathcal{H}$ and vertical $\mathcal{V}$  momentum 
 are $6$-dimensional and may be  presented in our notation as:
\begin{equation*}
\begin{split}
&\mathcal{H}=\{(k, u, (q,x))\in P \; : \;  x^2k +\langle k, q \rangle q - xq\times u=0 \quad \mbox{and} \quad (1-x^2)u-  
x k\times q - \langle u,q \rangle q=0 \}, \\
&\mathcal{V}=\{(k, u, (q,x))\in P \; : \; \langle u,q \rangle =0 \quad \mbox{and} \quad  k\times q +xu =0 \}.
\end{split}
\end{equation*}

\subsection{The axisymmetric 4D Chaplygin sphere}

We now continue working with our assumption that the sphere is axisymmetric and the mass tensor has the form~\eqref{eq:axisymm}.
In terms of the notation introduced above, Eq.~\eqref{eq:M-O-axi-symm} yields
\begin{equation*}
k=2\eta+b q\times (\eta \times q)-   bx \xi \times q, \qquad u=(2+a+b)\xi -b q\times (\xi \times q) +bx \eta \times q.
\end{equation*}
The above relations may be inverted, e.g. using Proposition~\ref{P:inversion}, to give
\begin{equation}
\label{eq:inverse-Leg}
\begin{split}
\xi &= \frac{b(2+b)(1-x^2)+\Delta(x)}{(a+b+2)\Delta(x)} u - \frac{b(2+b)\langle u,q \rangle }{(a+b+2)\Delta(x)}q - \frac{bx}{\Delta(x)} k\times q, \\
\eta &= \frac{(a+2+bx^2)}{\Delta(x)} k + \frac{(a+2)b\langle k,q \rangle}{2\Delta (x)} q + \frac{bx}{\Delta(x)} u \times q.
\end{split}
\end{equation}
Substitution of \eqref{eq:inverse-Leg} into \eqref{eq:motion-n4} gives  the equations of motion written in explicit form. As follows
from Proposition~\ref{P:Measure}, the resulting equations possess the invariant measure
\begin{equation*}
\mu = \frac{dk \,du \, dq  \, dx}{\Delta(x)}.
\end{equation*}

\subsection{Reduction by $\SO(3)$}

The $\SO(3)$  symmetry introduced in Section~\ref{SS:extra-symmetry} takes the following form in our notation.
The action of $h\in \SO(3)$ on a point $(k, u, (q,x))$ in phase space $P$ is 
\begin{equation*}
h\cdot (k, u, (q,x)) = (hk,hu,(hq,x)).
\end{equation*}
It is seen  from \eqref{eq:inverse-Leg} that $(\xi, \eta)$ transform to $(h\xi, h\eta)$ and hence, as predicted by the
discussion in Section~\ref{SS:extra-symmetry}, the equations of motion 
\eqref{eq:motion-n4} are equivariant. 

The action is not free since points in $P$ where  $k$, $u$ and  $q$ are parallel have a one-dimensional isotropy
subgroup isomorphic to $\SO(2)$. As a consequence, 
the reduced space $P/\SO(3)$  is  not smooth but is rather a stratified space.
To describe it  we recall 
 that the ring of  invariants  of the action is generated by the pairwise
inner products of the vectors $k$, $u$, $q$ and the triple vector product $\delta = \langle q,k\times u \rangle$, see e.g. \cite{Kraft}.
Let us define:
\begin{equation*}
A=\|k\|^2, \quad B=\| u\|^2, \quad E=\langle k,q\rangle, \quad G=\langle u,q \rangle.
\end{equation*}
We also recall that $\|q\|^2=1-x^2$ and  $F_1=\langle k,u \rangle$. These invariants are not independent but satisfy
\begin{equation*} 
\delta^2=\det(\Lambda) \quad \mbox{where} \quad \Lambda= \begin{pmatrix} A & F_1 & E \\  F_1 & B & G \\ E & G & 1-x^2 \end{pmatrix}.
\end{equation*}
Moreover they satisfy the inequalities
\begin{equation}
\label{eq:PrincMinors}
AB-F_1^2\geq 0, \qquad B(1-x^2)-G^2\geq 0, \qquad A(1-x^2)-E^2\geq 0.
\end{equation}

The reduced space $P/\SO(3)$ is isomorphic as a stratified space   to the 6 dimensional  
semi-algebraic variety $\mathcal{R}$ imbedded in $\R^7$
as 
\begin{equation*}
\mathcal{R}= \left \{ (A, B, F_1, E, G, x , \delta) \in \R^7 \, : \,  A, B \geq 0, \; -1\leq x \leq 1, \; \delta^2 -\det (\Lambda)=0  \;\; \mbox{and 
\eqref{eq:PrincMinors} hold}\;\;  \right \}.
\end{equation*}

As may be verified directly using Eqs.~\eqref{eq:motion-n4} and  \eqref{eq:inverse-Leg},  the reduced equations of motion are the restriction of the following vector field on $\R^7$ to $\mathcal{R}$: 
\begin{equation}
\label{eq:red-system}
\begin{split}
\dot A &= -\frac{2b(2+b)G\delta}{(a+b+2)\Delta(x)} + \frac{2bx}{\Delta(x)}(F_1E-GA), \\ 
\dot B & = \frac{2b(2+b)G\delta}{(a+b+2)\Delta(x)} - \frac{2bx}{\Delta(x)}(F_1E-GA), \\
\dot E & = -\frac{abx}{(a+b+2)\Delta(x)}GE- \frac{2x}{\Delta(x)}F_1, \\
\dot G  &=  \frac{bx}{\Delta(x)} \left (  A(1-x^2)   -E^2 + \frac{(2+b)}{2+a+b} G^2      \right )   -x \left (    \frac{b(2+b)(1-x^2)+\Delta(x)}{(a+b+2)\Delta(x)} \right) B + \frac{2+b-2bx^2}{\Delta(x)}\delta,
 \\
\dot F_1 & =0, \\ \dot x&=\frac{G}{a+b+2}, \\
\dot \delta &= \left ( \frac{bx^2(a+b+2)-\Delta(x)}{(a+b+2)\Delta(x)} \right ) AG -  \left ( \frac{b(2+b)(1-x^2)}{(a+b+2)\Delta(x)} \right ) BG
+\left ( \frac{2+b-2bx^2}{\Delta(x)} \right )F_1E  \\ & \qquad \qquad \qquad
+\frac{b(2+b)}{(a+b+2)\Delta(x)}G(G^2-E^2) - \frac{abx}{(a+b+2)\Delta(x)}G\delta.
\end{split}
\end{equation}

The 4, generically independent, integrals of the system are invariant under the action and descend to  
\begin{equation*}
\begin{split}
&F_1, \qquad F_2= A+B, \qquad F_3=x^2A+B(1-x^2)+E^2-G^2+2x\delta , \\
&H=\ \frac{a+2+bx^2}{2\Delta(x)}  A +  \frac{2+b(1-x^2)}{2\Delta(x)} B 
-   \frac{b(2+b)}{2(a+b+2)\Delta(x)} G^2 +  \frac{b(2+a)}{4\Delta(x)} E^2 +\frac{bx}{\Delta}\delta.
\end{split}
\end{equation*}

The  invariant measure also passes to the quotient. It is the restriction of the volume form  
\begin{equation*}
\mu_r= \frac{dA\, dB \, dE \, dG \, dF_1 \, dx \, d\delta}{\Delta(x)^2}.
\end{equation*}
on the ambient space $\R^7$ to $\mathcal{R}$.

Considering that the generic dimension of $P/\SO(3)$ is 6, and there exist  4 independent integrals and an invariant measure, the Euler-Jacobi  Theorem (see e.g. \cite{Arnold, BorMam2008}) implies that the reduced system is integrable. In particular,  the regular compact level sets of the integrals which have no equilibrium points are 2-tori where the flow is quasi-periodic 
after a time reparametrisation. Determining the nature of the  reconstruction of this type of dynamics to $P$ is a difficult problem for which little is known~\cite{Zung, FGNG}.   

On the other hand, the dynamics on the subsets $\mathcal{V}$ and $\mathcal{H}$ of $P$ may be completely 
described   in the light of Theorem~\ref{Th:Re} and the work of Jovanovi\'c~\cite{Jovan}, respectively. For completeness, we show how these
results may be deduced from the reduced system~ \eqref{eq:red-system}.
First note that $\mathcal{V}$ and $\mathcal{H}$ are $\SO(3)$-invariant and, 
under the symmetry reduction, project to subsets $\mathcal{V}/\SO(3)$ and $\mathcal{H}/\SO(3)$ of $\mathcal{R}$ which are invariant by the flow of  \eqref{eq:red-system}.
These are given by
\begin{equation*}
\mathcal{H}/\SO(3)= \{ (A, B, F_1, E, G, x , \delta) \in \mathcal{R} \, : \, F_1=0, \quad  E=0, \quad xA=-\delta, \quad (1-x^2)B-G^2=-x\delta \},
\end{equation*}
and
\begin{equation*}
\mathcal{V}/\SO(3)= \{ (A, B, F_1, E, G, x , \delta) \in \mathcal{R} \, : \,G=0, \quad  F_1=0, \quad xB=\delta, \quad (1-x^2)A-E^2=x\delta \}.
\end{equation*}

\subsection*{The  dynamics in the case of vertical momentum}

Theorem~\ref{Th:Re} guarantees that the set $\mathcal{V}/\SO(3)$ consists of equilibrium points. This  may be verified by substituting the relations in the description of $\mathcal{V}/\SO(3)$ into the reduced equations~\eqref{eq:red-system} and checking that the right hand side of the equations vanishes.
Considering that  $\SO(3)$ has rank 1, we 
conclude that the
$\mathcal{V}$ is foliated by periodic orbits.

\subsection*{The  dynamics in the case of horizontal momentum}

This is the case considered by Jovanovi\'c \cite{Jovan}, whose work implies  integrability for arbitrary $n$. 

We note that along this set the integrals $F_1$ and $F_3$  vanish.  
Now we restrict the flow  to the level set within $\mathcal{H}/\SO(3)$ of the other two integrals. Suppose that $H=h\geq 0$ and $F_2=f_2\geq 0$. We may write
\begin{equation}
\label{eq:A-on-J}
A=\frac{(2h(1-x^2)(a+b+2) -G^2)\Delta(x)}{(a+b+2)(2+a(1-x^2))} = (1-x^2)f_2 -G^2.
\end{equation}
The evolution equation for $G$ in \eqref{eq:red-system} may be simplified  using 
 the identities that define  $\mathcal{H}/\SO(3)$  together with  \eqref{eq:A-on-J} to eliminate the dependence on
 $\delta, B, A$ and $G^2$. Together with the equation for $x$ one  obtains the uncoupled  $2\times 2$ linear system
 \begin{equation*}
\dot G=- \frac{f_2-2bh}{2+a}x, \qquad \dot x= \frac{G}{a+b+2}.
\end{equation*}
The solution of this system with initial condition $x(0)=x_0$,  $G(0)=G_0$ is
 \begin{equation*}
x(t) = x_0 \cos \omega t, \qquad G(t) = -(a+b+2)x_0 \omega \sin \omega t,
\end{equation*}
where 
\begin{equation*}
\omega^2= \frac{f_2-2bh}{2+a}=\frac{aG_0^2+4h(a+b+2)}{(a+b+2)(2+a(1-x_0^2))} \geq 0.
\end{equation*}
The evolution of $A$, $B$ and $\delta$, is also periodic as may be seen respectively
 from \eqref{eq:A-on-J}, and the relations $B=f_2-A$, and $\delta=-xA$.
Therefore $\mathcal{H}/\SO(3)$  is foliated by periodic orbits. Since the rank of $\SO(3)$ is 1, the
  reconstructed motion on $\mathcal{H}$ consists of
quasi-periodic motion on 2-dimensional tori (see e.g. \cite{Field80}, \cite{Field-book}).
This is in agreement with \cite[Theorem 9]{Jovan}.

\subsection{A family of steady rotations}

Finally, we describe  another class of solutions of Eqs.~\eqref{eq:motion-n4} and  \eqref{eq:inverse-Leg} which have constant
angular velocity and lead to quasi-periodic dynamics on $P=\R^3\times \R^3 \times \Ss^3$. For this purpose, it is convenient to write Eqs.~\eqref{eq:motion-n4} and  \eqref{eq:inverse-Leg} in terms of  $\xi$,  $\eta$, $q$ and $x$, and without involving $k$ or $u$. After a long,
but straightforward calculation, one finds:
\begin{equation}
\label{eq:Motion-Velocity}
\begin{split}
&\dot \xi = \frac{a(2+b(1-x^2))}{\Delta(x)} \xi \times \eta - \frac{ab(b+2)\langle q,\xi\times \eta \rangle}{(2+a+b)\Delta (x)}  q, \qquad 
\dot \eta=-\frac{abx}{\Delta(x)}   q\times ( \xi\times \eta),  \\ 
 &\dot q  =q \times \eta -x\xi,  \qquad \dot x  = \langle q,\xi \rangle.
\end{split}
\end{equation}
\begin{proposition}
\label{P:steady-rot}
If the initial angular velocity  
\begin{equation*}
\Omega_0= \begin{pmatrix} \hat \eta_0 & \xi_0 \\ -\xi_0^T & 0 \end{pmatrix}\in \so(4),
\end{equation*}
of the 4D axisymmetric Chaplygin sphere satisfies $\eta_0 \times \xi_0=0$, then the angular velocity remains constant throughout the motion and the solution for $M(t)$ and $\gamma(t)$ with respective initial conditions $M_0$, $\gamma_0$ is 
given by
\begin{equation*}
M(t)= \exp (-\Omega_0 t) M_0  \exp (\Omega_0 t), \qquad  \gamma(t)= \exp (-\Omega_0 t) \gamma_0.
\end{equation*}
\end{proposition}
\begin{proof}
 It is readily seen from~\eqref{eq:Motion-Velocity} that if
$\xi$ and $\eta$ are parallel at some time,  then they remain  constant throughout the motion.
So, for these initial conditions the angular velocity $\Omega(t)= \Omega_0$  and we have a steady rotation. 
It is then straightforward to check that $M(t)$ and $\gamma(t)$ as defined above satisfy~\eqref{eq:motion-general}.
\end{proof}

Recall that  points in $P$ with non-trivial $\SO(3)$ isotropy are those for which  $k, u$ and $q$ are collinear.  In view of  Eq. \eqref{eq:inverse-Leg}  at these points the vectors $\xi$, $\eta$ and  $q$ are also parallel. Therefore, Proposition~\ref{P:steady-rot} describes the dynamics for this type of initial conditions.

We finish our discussion by indicating that the initial conditions in Proposition~\ref{P:steady-rot} do not generically  belong to
 $\mathcal{H}$
nor $\mathcal{V}$.

\section*{Conclusions and future work}

We considered the dynamics of the axisymmetric $n$-dimensional Chaplygin sphere. Our 
 main contribution  is to show that the dynamics 
is quasi-periodic when the angular momentum about the contact point is vertical. Also, to indicate how a further reduction of the system
by the additional $\SO(n-1)$ symmetry may be useful to understand the dynamics of the system for generic initial conditions. 
Indeed, in the case
$n=4$ the dynamics on the reduced system $P/\SO(3)$ was shown to be integrable by the Euler-Jacobi theorem.
The following   problems remain open:
\begin{enumerate}

\item[ $\bullet$] Determine if the equations  \eqref{eq:motion-general} on $P=\so(n)\times \Ss^{n-1}$ allow a Hamiltonisation.  Such 
Hamiltonisation is known to exist only in the case $n=3$ \cite{BorMamChap} and for initial conditions on $\mathcal{H}$ 
 for $n\geq 4$ \cite{Jovan}.

The existence of a Hamiltonian structure for the equations on $P$ would be useful to reconstruct the dynamics
from the reduced space $P/\SO(n-1)$. In particular, for $n=4$ the  (reparametrised) quasi-periodic dynamics on $P/\SO(3)$
predicted by Euler-Jacobi's Theorem would be guaranteed to be quasi-periodic on $P$ (see~\cite{Zung}).

\item[$\bullet$] Obtain the reduced equations on $P/\SO(n-1)$ for $n\geq 5$ and determine if they are integrable.

\end{enumerate}

\section*{Acknowledgements}

 I am grateful to the anonymous referees for  remarks that helped me to improve this paper.
I am very thankful to J. Montaldi for a conversation that inspired this research during my recent visit to the 
 University of Manchester. I  acknowledge support of the Alexander von Humboldt Foundation for a Georg Forster Experienced Researcher Fellowship that funded a research visit to TU Berlin where this work was done. Finally, I express my gratitude to A.V. Borisov for his invitation
to  submit a paper for the special issue of Regular and Chaotic Dynamics in the honour of S.A. Chaplygin on the occasion of his 150th anniversary.

\vskip 1cm

\noindent LGN: Departamento de Matem\'aticas y Mec\'anica, 
IIMAS-UNAM.
Apdo. Postal 20-126, Col. San \'Angel,
Mexico City, 01000,  Mexico. luis@mym.iimas.unam.mx.  \\

\end{document}